\documentclass[a4paper,USenglish,cleveref, autoref, thm-restate]{lipics-v2021}
\usepackage[vlined,ruled,linesnumbered]{algorithm2e}

\title{A Simplified  Parameterized Algorithm for Directed Feedback Vertex Set} 

\titlerunning{A Simplified  Parameterized Algorithm for DFVS} 

\author{Ziliang Xiong}{University of Electronic Science and Technology of China, Sichuan, China}{forgottencosecant@outlook.com}{https://orcid.org/0009-0005-2260-706X}{}

\author{Mingyu Xiao\footnote{The corresponding author.}}{University of Electronic Science and Technology of China, Sichuan, China}{ myxiao@gmail.com}{https://orcid.org/0000-0002-1012-2373}{National Natural Science Foundation of China, grant 62372095.}

\authorrunning{Z. Xiong and M. Xiao} 

\Copyright{} 



\category{} 

\relatedversion{} 

\DeclareMathOperator{\expand}{expand}
\newcommand{\ve}{\varepsilon}

\newtheorem{problem}[theorem]{Problem}

\begin{document}

\maketitle


\begin{abstract}
The Directed Feedback Vertex Set problem (DFVS) asks whether it is possible to remove at most $k$ vertices from a directed graph to make it acyclic. Whether DFVS is fixed-parameter tractable was a long-standing open problem in parameterized complexity until it was solved by Chen et al. in 2008 (STOC 2008). Now the running-time bound of this problem is improved to $\mathcal O(k!4^kk^5(n+m))$ (Lokshtanov et al, SODA 2018), where $n$ and $m$ are the numbers of vertices and arcs in the graph. In this paper, we simplify one crucial step in all previous parameterized algorithms for DFVS, which is to solve the compression version of the problem, and refine the running-time bound for DFVS to $\mathcal O(k!2^{o(k)}(n+m))$.
\end{abstract}


\section{Introduction}

For a directed graph, a vertex subset is called a \emph{directed feedback vertex set} (dfvs) if its removal makes the whole graph acyclic.
The  Directed Feedback Vertex Set problem (DFVS)
is to check whether a directed graph has a dfvs of size at most $k$.
This problem is among the first 21 NP-complete problems identified by Karp \cite{DBLP:conf/coco/Karp72}, and whether it is fixed-parameter tractable was a long-standing open problem in parameterized complexity \cite{DBLP:conf/coco/DowneyF92,DBLP:series/mcs/DowneyF99,DBLP:journals/jcss/GuoGHNW06,DBLP:conf/iwpec/KanjPS04,DBLP:conf/wads/RamanS03}. Finally, Chen et al. \cite{DBLP:conf/stoc/ChenLL08} showed that this problem can be solved in $\mathcal O(k!4^k k^3 n^4)$ time by combining two techniques, iterative compression and important separators.
Subsequently, attention has been paid to further improvements on the running time bound.
Especially, whether there exists a single-exponential time parameterized algorithm becomes another challenging open problem in parameterized algorithms. This problem is too hard and  Saurabh \cite{FDPC} even asked in `` Recent Advances in Parameterized Algorithms''  whether there exists an algorithm that solves DFVS in $\left(\frac kc\right)^k n^{\mathcal O(1)}$ time for some larger constant $c$. For previous results, the constant $c$ is $e/4$. In this paper, we improve the constant $c$ to $e/(1+\ve)$ for arbitrary constant $\ve > 0$.

To better understand the research history of DFVS, we quickly review the iterative compression technique. Fix an arbitrary order of the vertices of the graph. Let $G_i$ be the subgraph induced by the first $i$ vertices in the order. The idea is to solve the problem on the graph $G_i$ based on a solution to the graph $G_{i-1}$. Assume that we have found a dfvs $S_{i-1}$ in $G_{i-1}$ with size $|S_{i-1}|= k$. Then $S'=S_{i-1}\cup \{v_i\}$ is a dfvs in $G_{i}$
with size $|S'|=k+1$. Compared to the original DFVS problem on $G_i$, 
we have one more piece of information, a solution $S'$ of size $k+1$. So, we need to check whether the graph $G_i$ has a dfvs of size $k$ based on a given dfvs of size $k+1$. We call this problem the DFVS Compression problem.

\begin{problem}[The DFVS Compression Problem]
    Given a directed graph $G$, an integer $k$, and a dfvs $W$ of $G$ with size $|W|=k+1$, decide whether there exists a dfvs of $G$ of size at most $k$.
\end{problem}

To solve DFVS on the graph $G$, we only need to solve DFVS Compression on the graph $G_i$ for $i$ from 1 to $n$. This iterative compression method has become a standard method to solve DFVS, and most parameterized algorithms for DFVS adopt this step.

One framework to solve DFVS Compression is to first enumerate which vertices in $S'$ should be included in the solution $S_i$ to $G_i$. The set of vertices in $S' \cap S_i$ will be deleted directly and other vertices $S'\setminus S_i$ not in the solution should not be deleted. Thus, the problem can be converted to the Disjoint DFVS Compression problem, which is to find a dfvs that is disjoint from the given dfvs and also strictly smaller than the given dfvs.
Note that enumerating $S'\cap S_i$ will incur a factor of $2^k$ to the running time.
The Disjoint DFVS Compression problem can be solved in $\mathcal O(k!4^kn^3)$ time via enumerating important separators \cite{DBLP:journals/jacm/ChenLLOR08}, which dirctly implies that the DFVS Compression problem can be solved in $\mathcal O(k!8^kn^3)$ time.
Later, Cygan et al. \cite{DBLP:books/sp/CyganFKLMPPS15} illustrated a simpler representation of the algorithm in \cite{DBLP:journals/jacm/ChenLLOR08}, in which we do not need to enumerate $S'\cap S_i$,  and improve the running time bound to $\mathcal O(k!4^kk^4n(n+m))$ with careful analysis.

Another interesting contribution was due to Lokshtanov et al. \cite{DBLP:conf/soda/LokshtanovR018}. They proposed a new powerful framework named recursive compression, to replace the iterative compression framework.
As a result, this framework replaces a factor $n$ with the parameter $k$ in the time complexity of the iterative compression-based algorithms, implying that DFVS can be solved in $\mathcal O(k!4^kk^5(n+m))$ time. 
Theorem 3.1 in \cite{DBLP:conf/soda/LokshtanovR018} implies the following.

\begin{lemma}\label{thm:linear}
    If there is an algorithm solving the DFVS Compression problem in $\mathcal O(f(k)|W|(n+m))$ time, then there is an algorithm solving DFVS in $\mathcal O(f(k)k(n+m))$ time.
\end{lemma}

Therefore, a faster algorithm for the DFVS Compression problem implies a faster algorithm for the DFVS problem. By the above lemma, we can simply focus on the 
DFVS Compression problem.
In this paper, we not only greatly simplify previous algorithms for the DFVS compression problem but also improve the running time bound to $\mathcal O(k!2^{o(k)}(n+m))$. Therefore, we get the following result.

\begin{theorem} \label{thm:time}
DFVS can be solved in $\mathcal O(k!2^{o(k)}(n+m))$ time.
\end{theorem}

Next, before going to the details of our algorithm and analysis, we first briefly introduce the main idea of our algorithm for the DFVS Compression problem and why we can achieve the improvements.

Essentially, previous algorithms would guess a permutation $\pi$ of the vertex set $W$, which will generate a factor of $\mathcal O(k!)$ in the running time, and then go over each element of $W$ in order of the guessed permutation, and guess all the $\mathcal O(4^k)$ important separators for each element.
The total running time will be $\mathcal O(k!4^k)$.

In this paper, we observe that the guess of the important separator for the last element in $W$ does not depend on the permutation $Pi$. Thus, instead of guessing $\pi$ first and then the important separators, we guess the last element $w$ in $\pi$ and then the important separator from $w$ to other vertices in $W$. The benefit of this is one additional reduction rule: if a vertex $w' \in W$ is not in any directed cycles then we remove simply $w'$ from the set  $W$. This simple reduction, however, is useful in the analysis.

The main idea of the analysis goes roughly as follows.
If the important separator of the last element we guessed is always of size 1,  then we will get the running time $O(k!)$, which is equal to the time of enumerating all permutations of $W$.
Whenever some important separator has size $i\geq 2$, then it contributes about $4^i$ important separators and a factor of $4^i$ in the running time. However, it also ensures that the above reduction rule will be applied for at least $i-1$ times since the sum of all important separator sizes is bounded by $k$. 
If the reduction rule is applied more than $z = c\cdot k/\log k$ times for a sufficiently large constant $c$, then we end up saving a factor of $z! > 4^k$ in the running time. Thereby we will win over the factor incurred by all of the important separator guessings and again can get the same bound $4^k$.
Thus, the reduction rule is applied for at most $c\cdot k/\log k$ times. For this case, the sum of all the important separators of size $>1$ is $\mathcal O(k/\log k)$. So the factor for all of the important separator branchings is at most $2^{o(k)}$. Overall, we will get a bound of $\mathcal O(k!2^{o(k)})$.

\section{Preliminaries}
We use angle brackets to denote a sequence of elements,
e.g., $\langle a, b, c \rangle$ denotes a sequence of $a,b$, and $c$. We also use $\langle\rangle$ to denote an empty sequence.
Let $G$ be a directed graph with $n$ vertices and $m$ arcs.
For a directed graph $G'$, the set of vertices and the set of arcs are denoted by $V(G')$ and $E(G')$, respectively.
The subgraph of $G$ \emph{induced} by a vertex subset $X \subseteq V(G)$ is the graph obtained by removing all vertices not in $X$ and their incident arcs, denoted by $G[X]$.
A \emph{path} of $G$ is a sequence of vertices in which there is an arc pointing from each vertex in the sequence to its successor in the sequence.
We say a path $P$ passes through an arc $(u,v)$ if $v$ is a successor of $u$ in $P$, and use $E(P)$ to denote the arcs that $P$ passes through.
For two disjoint vertex sets $X$ and $Y$, an $(X,Y)$-path of $G$ is a path that starts with a vertex in $X$ and ends with a vertex in $Y$.
We say a vertex $r$ is reachable from $X$ if there exists an $(X,\{r\})$-path in $G$.
A cycle of $G$ is a path of $G$ that starts and ends at the same vertex.
A dfvs of $G$ is a set of vertices whose removal from $G$ results in an acyclic directed graph.
A dfas of $G$ is a set of arcs whose removal from $G$ results in an acyclic directed graph.
An $(X,Y)$-cut of $G$ is a set of arcs whose removal breaks every $(X,Y)$-path of $G$.
An $(X,Y)$-cut $F$ is \emph{important} if it satisfies the following two conditions.

\begin{enumerate}
    \item
    $F$ is inclusion-wise minimal. That is, any proper subset $F' \subset F$ is not an $(X,Y)$-cut.
    \item
    Let $R$ be the set of vertices reachable from $X$ in $G- F$.
    For any $(X,Y)$-cut $F'$ that is not larger than $F$, the set of vertices reachable from $X$ in $G - F'$ is a subset of $R$.
\end{enumerate}

We use the following properties of important cuts in this paper.

\begin{proposition}[Proposition 8.35 of \cite{DBLP:books/sp/CyganFKLMPPS15}]\label{impcut-exist}
    Let $G$ be a directed graph and $X,Y \subseteq V(G)$ be two disjoint sets of vertices.
Let $S$ be an $(X,Y)$-cut. There is an important $(X,Y)$-cut $S'$ (possibly, $S'=S$) such that $|S'| \leq |S|$ and every vertex reachable from $X$ in $G - S$ is also reachable from $X$ in $G-S'$.
    
\end{proposition}

\begin{theorem}[Theorem 8.36 of \cite{DBLP:books/sp/CyganFKLMPPS15}]\label{impcut-enum}
    Let $G$ be a directed graph with $n$ vertices and $m$ arcs, and $X,Y \subseteq V(G)$ be two disjoint sets of vertices, let $k \geq 0$ be an integer, and let $\mathcal S_k$ be the set of all important $(X,Y)$-cuts of size at most $k$.
    Then $|\mathcal S_k| \leq 4^k$ and $\mathcal S_k$ can be construted in time $\mathcal O(|\mathcal S_k|k^2(n+m))$.
\end{theorem}

\begin{lemma}[Lemma 8.37 of \cite{DBLP:books/sp/CyganFKLMPPS15}]\label{size1}
    Let $G$ be a directed graph with $n$ vertices and $m$ arcs, and $X,Y \subseteq V(G)$ be two disjoint sets of vertices.
    If $\mathcal S$ is the set of all important $(X,Y)$-cuts, then $\sum_{S \in \mathcal S} 4^{-|S|} \leq 1$.
\end{lemma}

\begin{lemma}
    For any two disjoint sets of vertices $X$ and $Y$ in a directed graph, there is at most one important $(X,Y)$-cut of size one.
\end{lemma}

\begin{proof}
Assume to the contrary that there are  two different important $(X,Y)$-cuts $\{e_1\}$ and $\{e_2\}$ of size one.   Let $R_1$ and $R_2$ be the sets of vertices reachable from $X$ in $G-\{e_1\}$ and  $G-\{e_2\}$, respectively.
    Since the size of the minimum $(X,Y)$-cut is 1 and both $\{e_1\}$ and $\{e_2\}$ are $(X,Y)$-cuts, there is an $(X,Y)$-path $P$ that passes through $e_1$ and $e_2$. Without loss of generality, we assume that $e_1$ precedes $e_2$ in $P$. Then we have $R_1 \subset R_2$, a contradiction with $\{e_1\}$ being an important cut.
\end{proof}




\section{Solving DFVS Compression}
Now we are ready to solve the DFVS Compression problem. We will follow the framework in \cite{DBLP:books/sp/CyganFKLMPPS15}:
Instead of solving DFVS Compression directly, we convert the problem to the following feedback arc set problem in a directed graph.  

\begin{problem}[Directed Feedback Arc Set Problem with dfvs as Hint (DFAS-V)]
    Given a directed graph $G$, an integer $k$, and a dfvs $W$ of $G$ with size $|W| = k+1$, decide whether there exists a dfas $F$ of $G$ whose size is at most $k$.
\end{problem}

\begin{definition}\label{def-expand}
For a directed graph $G$ and a vertex $w \in V(G)$, let $\expand(G,w,(w^-,w^+))$ be the directed graph obtained by splitting $w$ into two vertices $w^-$ and $w^+$ connected by an arc $(w^-,w^+)$, where all arcs pointing to $w$ will be pointing to $w^-$, and all arcs starting from $w$ will be starting from $w^+$.
\end{definition}

Let $I=(G,W,k)$ be an instance of the DFVS Compression problem.
We construct an instance $I'=(G',W',k)$ of the DFAS-V problem that is equivalent to $I$.
Let $G'$ be the directed graph obtained from $G$ by iteratively calling $\expand(G,v_i,(v_i^-,v_i^+))$ for all vertices $v_i\in V(G)$, i.e., $G'$ is obtained from $G$ by expanding every vertex $v_i\in V(G)$ into $(v_i^-, v_i^+)$.
Let  $W'=\{ v_i^+ | v_i \in W\}$. 
The following lemma shows the equivalency between $I$ and $I'$.

\begin{lemma}\label{sec-f}
    An instance $(G,W,k)$ is a yes-instance of the DFVS Compression problem if and only if $(G',W',k)$ is a yes-instance of DFAS-V.
    Furthermore, given a solution $S'$ to $(G',W',k)$, we can construct a solution $S$ to $(G,W,k)$ in time linear to the size of $G$.
\end{lemma}

\begin{proof}
    Let $S$ be a solution to the DFVS Compression instance $(G,W,k)$.
    We construct $S'=\{(v^-,v^+) | v \in S\}$ and argue that $S'$ is a solution to the DFAS-V instance $(G',W',k)$.
    By Definition \ref{def-expand}, we know that arc $(v^-,v^+)$ is the only arc that ends at $v^+$ and the only arc that starts from $v^-$.
    Thus for every cycle $C'$ in $G'$, there is a cycle $C$ in $G$ that passes through the arc set $\{(u,v) | (u^+,v^-)\in E(C')\}$, and the vertices it contains is $\{v | (v^-,v^+) \in E(C') \}$.
    Moreover, there is a vertex $v \in C \cap S$ since $S$ is a dfvs, and we have $(v^-,v^+) \in S' \cap E(C')$ by the definition of $S'$.
    Since $|S| \leq k$, we have $|S'| \leq k$ thus $S'$ is a solution of $(G',W',k)$.
    
    Let $S'$ be a solution to the DFAS-V instance $(G',W',k)$.
    Let $S'' = \{(v^-,v^+) | (v^-,v^+) \in S' \vee (u^+,v^-) \in S'\}$.
    We state that $S''$ is still a dfas in $G'$ because every cycle that passes through $(u^+,v^-)$ also passes through $(v^-,v^+)$, the only outgoing arc of $v^-$.
    Next, we construct $S = \{v | (v^-,v^+) \in S''\}$ and argue that $S$ is a solution to the DFVS Compression instance $(G,W,k)$.
    For every cycle $C$ in $G$, there is a cycle $C'$ in $G'$ that $E(C')=\{(u^+, v^-) | (u,v) \in E(C)\} \cup \{(v^-,v^+) | v \in C\}$.
    Moreover, there is an arc $(v^-,v^+) \in S'' \cap C'$ since $S''$ is a dfas, and we have $v \in S \cap C$ by the definition of $S$.
\end{proof}

Based on Lemma~\ref{sec-f}, we only need to solve DFAS-V.

\section{Solving DFAS-V}
In this section, we give our algorithm for DFAS-V.
Let $(G, W, k)$ be an instance of DFAS-V.
Our algorithm is based on the following two simple observations. 
\begin{observation}\label{onev}
    Let $G$ be a directed graph, $W$ be a dfvs of $G$, and $F$ be a dfas of $G$.
    There exists a vertex $w \in W$ such that there is no $(\{w\},W\setminus w)$-path in $G - F$.
\end{observation}

\begin{proof}
    Let $\pi$ be an arbitrary topological vertex order of $G- F$ and $w\in W$ be the vertex with the largest rank index under $\pi$.
    There is no $(\{w\},W\setminus w)$-path in $G- F$ by the definition of topological order.
\end{proof}

\begin{lemma}\label{alg_crucial}
    Let $G$ be a directed graph and $W$ be a dfvs of $G$.
    There exists a vertex $w \in W$ such that in $G'=\expand(G,w,(w^-,w^+))$ there is a minimum dfas containing an important $(\{w^+\}, (W\setminus w)\cup\{w^-\})$-cut.
\end{lemma}

\begin{proof}
    Let $F$ be a minimum dfas of $G$.
    By Lemma~\ref{onev}, we know that there is a vertex $w \in W$ such that there is no $(\{w\},W\setminus w)$-path in $G - F$.
        We note the following properties:
    \begin{enumerate}
        \item For any $(\{w\}, W \setminus w)$-path in $G$, there is a $(\{w^+\}, W \setminus w)$-path of the same arc set in $G'$, and vice versa.
        \item For any cycle that passes through $w$ in $G$, there is a $(\{w^+\}, \{w^-\})$-path of the same arc set in $G'$, and vice versa.
    \end{enumerate}
    Let $Z=(W\setminus w)\cup\{w^-\}$, $R$ be the set of vertices reachable from $w^+$ in $G' - F$, and $S$ be the set of arcs pointing out from $R$.
    If $(W \setminus w) \cap R \neq \emptyset$, by the above properties, there is a $(\{w\}, W \setminus w)$-path in $G-F$.
    If $w^- \in R$, then $F$ is not a dfas since $w^-$ is reachable from $w^+$ in $G'-F$, which implies a cycle that passes through $w$ in $G-F$.
    Therefore, $S$ intersects with every $(\{w^+\}, Z)$-path in $G'$.
    
    By Proposition \ref{impcut-exist}, there is an important $(\{w^+\}, Z)$-cut $S'$ in $G'$ such that $|S'|\leq |S|$ and $R \subseteq R'$, where $R'$ is the set of vertices reachable from $w^+$ in $G'- S'$.
    Let $F'=(F\setminus S) \cup S'$.
    Obviously, $F'$ is no larger than $F$ and $F'$ includes an important $(\{w^+\}, Z)$-cut $S'$.
    Now we argue that $F'$ is a dfas.
    Let $C$ be a cycle of $G$.
    If $C$ intersects with $S$, then $C$ also intersects with $S'$ since $C$ includes at least one vertex in $(R\setminus w) \cup \{w^+\}$.
    Otherwise, $C$ must intersect with $F\setminus S$ since $F$ is a dfas.
    Therefore $F'$ is a dfas of $G$ no larger than $F$, and $F'$ contains an important $(\{w^+\}, Z)$-cut.
\end{proof}

Based on Lemma~\ref{alg_crucial}, we can design algorithm by iteratively adding important $(\{w^+\}, Z)$-cuts to the solution set.
However, we do not know which vertex in $W$ is the vertex $w$ and do not know which important cut is the correct one.
We just enumerate all possible cases.
The whole algorithm is presented in Algorithm~\ref{alg:dfasv}.
The first seven steps handles trivial cases.
Step 8 is to enumerate all vertices in $W$ by taking each of them as the vertex $w$.
Step 10 is to consider each important cut of size at most $k$.
The algorithm is quite simple.
However, we analyze it in a neat way to get a better running time bound. 


\begin{algorithm}[t]
    \caption{DFAS-V$(G,W,k)$}
    \label{alg:dfasv}
    \KwIn{A directed graph $G$, a dfvs $W$ of $G$, and parameter $k$.}
    \KwOut{A dfas $F$ of $G$ whose size is at most $k$, or $nil$ if such set does not exist.}
    \uIf{$G$ is acyclic} {
        \Return{$\emptyset$}
    }
    \uElseIf {$k=0$} {
        \Return{nil}
    }
    \Else {
        \While{$\exists w \in W$ such that no cycle in $G$ contains $w$} {
            $W := W\setminus w$\\
        }
        \ForEach{$w \in W$}{
            $G' := \expand(G,w,(w^-,w^+))$\\
            \ForEach{important $(\{w^+\},(W\setminus w)\cup \{w^-\})$-cut $S$ in $G'$ with $|S| \leq k$}{
                $F := \text{DFAS-V}(G - S,W\setminus w,k-|S|)$\\
                \If {$F \neq nil$ } {
                    \Return {$F \cup S$}
                }
            }
        }
        \Return{nil}
    }
\end{algorithm}


\subsection{The Analysis}
Next, we analyze the running time bound of the algorithm in Algorithm~\ref{alg:dfasv} with the constraint that $|W|\leq k+1$ in the input. We derive the bound by analyzing the search tree generated by Algorithm~\ref{alg:dfasv}.
There are two branching operations.
The first is to consider each vertex in $W$ in Step 8 and the second is to consider each important cut with size at most $k$ in Step 10.
We call an important cut \emph{small} if it contains only one arc, we call it \emph{large} otherwise.
In fact, in Step 10, when we include a large important cut to the solution, one part of the running-time bound will decrease; when we include a small important cut to the solution, another part of the running-time bound will decrease.

We consider the trade-off between these two cases and get the following lemma to prove the time complexity of our algorithm.

\begin{lemma}
    \label{thm:num_leaves}
    Assume that $|W|\leq k+1$. For any constant $\ve \in (0,\frac 12)$, the number of leaves in the search tree generated by DFAS-V$(G,W,k)$ is bounded by $\mathcal O(k! k^2 \gamma_\ve^k)$, where $\gamma_\ve=\ve^{-\ve} (1-\ve)^{-(1-\ve)}$.
\end{lemma}

\begin{proof}
Let $T$ be the search tree generated by the algorithm. 
We consider the path from the root of $T$ to each leaf of $T$.
We categorize the path by the \textit{vertex sequence} $\pi = \langle w_1, w_2, \ldots, w_l \rangle$ and the \textit{cut-size sequence} $\mathcal{X} = \langle x_1, x_2, \ldots, x_l \rangle$,
where $w_i\in W$ is the vertex selected in Step 8 in the $i$th loop of the algorithm and $x_i$ is the size of the corresponding important cut in Step 10 in the $i$th loop. Therefore, we have $x_i\geq 1$ for each $i$ (because we have executed Steps 6-7) and $\sum_{i=1}^l x_i\leq k$. Note that $l$ can be a number smaller than $|W|$.

For each pair of fixed $\pi$ and $\mathcal{X}$, there might be more than one corresponding path because there may be multiple different important cuts of the same size $x_i$.
Let $B \subseteq \pi$ be a subset of vertices in $\pi$.
We use $f(\pi,B,k)$ to denote the number of paths satisfying the following properties: 
\begin{enumerate}
    \item its vertex sequence is $\pi$.
    \item the sum of the corresponding cut-size sequence is $k$, i.e., $\sum_{i=1}^l x_i= k$.
        \item for each $w_i\in B$, the corresponding cut-size $x_i$ is 1.
\end{enumerate}

This means if $w_i\in B$, then the corresponding important cut in the $i$-th loop is a small cut with size 1.
Note that for $w_j\notin B$, there is no requirement on $x_j$ and the cut can be either large or small. 
First, we prove the following claim.

\begin{claim}\label{claim}
For every $\pi, B, k$ we have $f(\pi,B,k)\leq 4^{k-|B|}$.
\end{claim}

We prove it by induction on $k$. 
When $k=0$, we have $\pi=\langle \rangle$ and $B=\emptyset$.
Therefore $f(\pi,B,0)=1$.
Assume that Claim \ref{claim} holds for the parameter being at most $k-1$. We can see that it also holds for the parameter being $k$. 

First, we consider the case that $w_1 \notin B$.
Let $\mathcal S$ be the set of important cuts with size $\geq 2$ corresponding to $w_1$. By Lemma~\ref{size1}, we have $\sum_{S \in \mathcal S}4^{-|S|}\leq 1$.
Therefore,

    $$
    f(\pi = \langle w_1, w_2, \ldots, w_l \rangle,B,k)
    =\sum_{S \in \mathcal S}f(\langle w_2, \ldots, w_l\rangle, B, k-|S|)$$
    $$
    =\sum_{S \in \mathcal S}4^{k-|S|-|B|}
    =4^{k-|B|}\sum_{S \in \mathcal S}4^{-|S|} \leq 4^{k-|B|}.
    $$
    
Second, we assume that $w_1 \in B$.
For this case, we have 
    
    $$
    f(\pi, B, k)
    =f(\langle w_2, \ldots, w_l \rangle, B-w_1, k-1)
    =4^{k-1-(|B|-1)} = 4^{k-|B|}.
    $$

Note that as the length of $\pi$ increases, so too should the size of $B$.
   The following observation shows some lower bound on the number of small cuts for each $\pi$, which is also a lower bound for $|B|$.
    Let $|\pi|=a$ and $|B|=b$,
    Since every \emph{large} important cut contains at least $2$ arcs, we have $k-b \geq 2(a-b)$.
    Therefore $|B|=b \geq 2a-k$.

    Let $g(k)$ be the number of leaves in the search tree generated by our algorithm.
    We are ready to analyze the value of $g(k)$ now.
    Let $\mathrm{perm}(A)$ be the set of all permutations of the set $A$.
    According to the definition of $f(\pi, B, k)$, we have that
    
    $$
    g(k) \leq \sum_{A \subseteq W} \sum_{\pi \in \mathrm{perm}(A)} \sum_{B \subseteq A} f(\pi,B,k).
    $$
    
    We use ``$\leq$'' instead of ``$=$'' in the above relation because we allow the important cuts corresponding to vertices in $A \setminus B$ also to be small.
    Let $|A|=a$ and $|B|=b$.
    Since $|W|\leq k+1$, the number of choices of $\pi$ with length $a$ is at most $(k+1)!/(k+1-a)!$. We have showed that the size of $B$
    can be $\max(2a-k, 0)$. The number of subsets $B$ of $A$ with size $\max(2a-k, 0)$ is $\binom{a}{\max(0,2a-k)}$. By Claim~\ref{claim}, we have that 
    $$
    g(k)\leq \sum_{a=1}^{k}\frac{(k+1)!}{(k+1-a)!} \binom{a}{\max(0,2a-k)} 4^{k-(2a-k)}.
    $$
    
    Let 
    $$
    h(k,a)=\frac{4^{2(k-a)}}{(k+1-a)!} \binom{a}{\max(0,2a-k)}. 
    $$
    
    Next we analyze $h(k,a)$, by distinguish several different relations between $a$ and $k$.
    Let $\ve$ be an arbitrary constant in $(0, \frac 12)$.
    Observe that $c^n/n! \leq c^c$ holds for arbitrary $c,n \geq 1$, we can use the denominator $(k-a+1)!$ to consume every factor in the form of $c^{k-a}$ in the nominator.
    
    When $a \leq k-a$, we have $2a-k \leq 0$. Therefore,
    
    $$
    h(k,a)=\frac{4^{2(k-a)}}{(k-a+1)!} \binom a0=\mathcal O(1).
    $$
    
    When $\ve a \leq k-a < a$, we have $a \leq (k-a)/\ve$. Therefore,
    
    $$
    h(k,a)\leq \frac{4^{2(k-a)}}{(k-a+1)!} 2^{(k-a)/\ve}  =\mathcal O(1).
    $$
    
    When $0 \leq k-a < \ve a$, we have
    
    $$
    h(k,a) \leq \frac{4^{2(k-a)}}{(k-a+1)!} \binom{a}{k-a} \leq O\left(\binom{a}{\ve a}\right) = \mathcal O(\gamma_\ve^a),
    $$
    where $\gamma_\ve=\ve^{-\ve} (1-\ve)^{-(1-\ve)}$.
    
    For each $a$, we have that $h(k,a)= \mathcal O(\gamma_\ve^a)$.
    Since $a \leq k$, we get 
    
    $$
    g(k) \leq (k+1)! \sum_{a=1}^{k}h(k,a)=\mathcal O \left((k+1)!\sum_{a=1}^{k}\gamma_\ve^a \right)
    =\mathcal O((k+1)!k \gamma_\ve^k)=\mathcal O(k!k^2\gamma_\ve^k).
    $$
\end{proof}

\begin{theorem}
DFAS-V$(G,W,k)$ with $|W|\leq k+1$ runs in $\mathcal O(k!2^{o(k)}(n + m))$ time.
\end{theorem}

\begin{proof}
    By Lemma \ref{thm:num_leaves}, we know that the number of nodes in the search tree is $\mathcal O(k!k^3\gamma_\ve^k)$ since the length of any search path is at most $k$.
    By Theorem \ref{impcut-enum}, the set $\mathcal S$ of all important $(X,Y)$-cuts of size at most $k$ in a directed graph $G=(V,E)$ can be constructed in $\mathcal O(|\mathcal S|k^2(n+m))$.
    Since each important cut found by the algorithm corresponds to an edge in the search tree, the number of important cuts found in the algorithm is the number of nodes in the search tree minus 1 and the time used to find important cuts in Alg.\ref{alg:dfasv} is $\mathcal O(k!k^5\gamma_\ve^k(n+m))$, w.r.t. constant $\ve \in (0, \frac 12)$.
    This is also the time complexity of the whole algorithm since the running time of every other step is linear to the size of the input graph.
    Specifically, we can use Tarjan's algorithm \cite{DBLP:journals/siamcomp/Tarjan72} to implement steps 6-7.
    When $\ve \in (0,\frac 12)$, the value $\gamma_\ve$ monotonically increases w.r.t. $\ve$ since $\frac{\mathrm d}{\mathrm d \ve} \left( \frac1{\gamma_\ve}\right) = \ve^\ve (1-\ve)^{1-\ve} (\log \ve - \log(1-\ve)) < 0$.
    Moreover, we have $\lim_{\ve \to 0} \gamma_\ve = 1$.
    Thus, for any constant $c>1$, there exists $\ve \in (0, \frac 12)$ such that $1<\gamma_\ve < c$, and the time complexity of Alg.\ref{alg:dfasv} is $\mathcal O(k!k^5 2^{o(k)}(n+m)) = \mathcal O(k!2^{o(k)}(n+m))$.
\end{proof}

Finally, by Lemma \ref{sec-f}, there is an algorithm that solves DFVS Compression in $\mathcal O(k!2^{o(k)}(n + m))$ time, and by Lemma \ref{thm:linear}, there is an algorithm that solves DFVS in $\mathcal O(k!2^{o(k)}(n + m))$ time, which completes the proof of Theorem \ref{thm:time}.

\section{Conclusion}
In this paper, we give a simplified algorithm for the DFVS Compression problem and improve the running-time bound for DFVS. The dominating part of the running-time bound is still $k!$.
We cannot avoid this because we indeed still need to enumerate all permutations of $k$ elements in the worst case. However, we cleverly combine enumeration of permutations and important cuts to simplify previous algorithms and improve the running-time bound. Now DFVS can be solved in $\mathcal O(\left(\frac kc\right)^k (n+m))$ time with $c=e/(1+\ve)$  for any constant $\ve > 0$.
Whether DFVS allows a single-exponential parameterized algorithm is still a challenging open problem in parameterized complexity.

\section*{Acknowledgments}
We thank all anonymous reviewers for their insight comments.
The work is supported by the National Natural Science Foundation of China, under grant 62372095.

\bibliography{pdfvs}
\end{document}